\pdfoutput 1
\def\draft{0} % change to 0 to turn off author notes

\documentclass[a4paper,10pt]{article}

\usepackage{amsmath,amsfonts,amssymb,amscd,amsthm,xspace}
\usepackage[colorlinks=true, urlcolor=blue, citecolor=blue, linkcolor=blue, filecolor=blue]{hyperref}
\hypersetup{%
  pdftitle = {On Extractors and Exposure-Resilient Functions for Sublogarithmic Entropy},
  pdfkeywords = {pseudorandomness, exposure-resilient function, randomness extractor, bit-fixing source},
  pdfauthor = {Yakir Reshef, Salil Vadhan}
}

\usepackage[square, numbers, comma, sort&compress]{natbib}

\ifnum\draft=1
\newcommand{\Snote}[1]{\textbf{[Salil's Note: #1]}}
\newcommand{\Ynote}[1]{\textbf{[Yakir's Note: #1]}}
\else
\newcommand{\Snote}[1]{}
\newcommand{\Ynote}[1]{}
\fi

\newtheorem{theorem}{Theorem}
\newtheorem{lemma}[theorem]{Lemma}
\newtheorem*{lemma*}{Lemma}

\newtheorem{claim}[theorem]{Claim}
\newtheorem*{claim*}{Claim}
\newtheorem{proposition}[theorem]{Proposition}

\theoremstyle{definition}
\newtheorem{definition}[theorem]{Definition}
\newtheorem{question}{Open Question}

\theoremstyle{remark}

\numberwithin{equation}{section}
\numberwithin{theorem}{section}

\newcommand{\Z}[1]{\mathbb{Z}/#1\mathbb{Z}}
\newcommand{\N}{\mathbb{N}}
\newcommand{\B}{\{0,1\}}

\DeclareMathOperator*{\Exp}{E}
\newcommand{\ep}{\varepsilon}

\newcommand{\ii}{^\prime}
\newcommand{\pr}[1]{#1}
\newcommand{\Un}[1]{\pr{U_{#1}}}
\newcommand{\dpr}[2]{\left(#1, #2 \right)}
\newcommand{\subsets}[2]{\genfrac{\{}{\}}{0pt}{}{#1}{#2}}

\title{On Extractors and Exposure-Resilient Functions for Sublogarithmic Entropy\thanks{Some of these results previously
appeared in the first author's undergraduate thesis~\citep{Reshef}.}
\ifnum\draft=1\\{\small \sc Working Draft: Please Do Not Distribute}\fi
}

\author{Yakir Reshef\thanks{Department of Mathematics, Harvard College. \href{mailto:yreshef@post.harvard.edu}{\nolinkurl{yreshef@post.harvard.edu}}.} \and
Salil Vadhan\thanks{School of Engineering and Applied Science, Harvard University, 33 Oxford Street, Cambridge, MA 02138. \href{mailto:salil@seas.harvard.edu}{\nolinkurl{salil@seas.harvard.edu}}.
\url{http://seas.harvard.edu/\~salil}.  Supported by US-Israel BSF grant 2006060
and NSF grant CNS-0831289.}}

\begin{document}

\maketitle

\begin{abstract}
We study deterministic extractors for oblivious bit-fixing sources (a.k.a. resilient functions) and exposure-resilient functions with small min-entropy: of the function's $n$ input bits, $k\ll n$ bits are uniformly random and unknown to the adversary.

We simplify and improve an explicit construction of extractors for bit-fixing sources with sublogarithmic $k$ due to Kamp and Zuckerman (SICOMP 2006), achieving error exponentially small in $k$ rather than polynomially small in $k$. Our main result is that when $k$ is sublogarithmic in $n$, the short output length of this construction ($O(\log{k})$ output bits) is optimal for extractors computable by a large class of space-bounded streaming algorithms.

Next, we show that a random function is an extractor for oblivious bit-fixing sources with high probability if and only if $k$ is superlogarithmic in $n$, suggesting that our main result may apply more generally. In contrast, we show that a random function is a static (resp. adaptive) exposure-resilient function with high probability even if $k$ is as small as a constant (resp. $\log\log n$). No explicit exposure-resilient functions achieving these parameters are known.

\end{abstract}

\vspace{0.15in}

\noindent {\bf Keywords:} pseudorandomness, exposure-resilient function, randomness extractor, bit-fixing source

\section{Introduction}

Randomness extractors are functions that extract almost-uniform bits from weak sources of randomness (which may have biases and/or correlations).  Extractors can be used for simulating randomized algorithms and protocols with weak sources of randomness, have close connections to many other ``pseudorandom objects'' (such as expander graphs and error-correcting codes), and have a variety of other applications in theoretical computer science.

The most extensively studied type of extractor is the {\em seeded extractor}, introduced
by Nisan and Zuckerman~\citep{NisanZu96}.  These extractors are given as additional input a small
``seed'' of truly random bits to use as a catalyst for the randomness extraction, and this allows for extracting almost-uniform bits from very unstructured sources, where all we know is a lower bound on the min-entropy.  In many applications, such as randomized algorithms, the need for truly random bits can be eliminated by trying all possible seeds and combining the results (e.g. by majority vote).
However, prior to the Nisan--Zuckerman notion,
there was a substantial interest in {\em deterministic extractors} (which have no random seed) for restricted classes of sources.  Over the past decade, there has been a resurgence in the study of deterministic extractors, motivated by settings where enumerating all possible seeds does not work (e.g. distributed protocols) and by other
applications in cryptography.

In this paper, we study one of the most basic models: an {\em oblivious bit-fixing source (OBFS)} is an $n$-bit source where some $n-k$ bits are fixed arbitrarily and the remaining $k$ bits are uniformly random.  Deterministic extractors for OBFSs, also known as {\em resilient functions (RFs)}, were first studied in the mid-80's, motivated by cryptographic applications~\citep{Vazirani87a,BennettBrRo88,ChorGoHaFrRuSm85}.  A more relaxed notion is that of an {\em exposure-resilient function (ERF)}, introduced in 2000
by Canetti et al.~\citep{CanettiEtAl}.  Here
all $n$ bits of the source are chosen uniformly at random, but $n-k$ of them are seen by an adversary; an ERF should extract bits that are almost-uniform even conditioned on what the adversary sees.  ERFs come in two types: {\em static} ERFs, where the adversary decides which $n-k$ bits to see in advance, and {\em adaptive} ERFs, where the adversary reads the $n-k$ bits adaptively.
In recent years, there has been substantial progress in giving explicit constructions of both
RFs and ERFs~\citep{CanettiEtAl,DodisSahaiSmith,KampZuckerman,GabizonRazShaltiel}.

In this paper, we focus on the case when $k$, the number of random bits unknown to the adversary, is very small, e.g. $k<\log n$.   While this case is not directly motivated by applications, it is interesting from a theoretical perspective for a couple of reasons:
\begin{itemize}
\item For many other natural classes of sources (several independent sources~\citep{ChorGoldreich88}, samplable sources~\citep{TrevisanVa00}, and affine sources~\citep{BarakKiShSuWi05}), at least logarithmic min-entropy is necessary for extraction.\footnote{For the case of 2 independent sources, the need for logarithmic min-entropy is proven in \citep{ChorGoldreich88}.  For sources samplable by circuits of size $s=n^2$, it can be shown by noting that the uniform distribution on any $2^k$ elements of $\{0,1\}^{k+1}\circ 0^{n-k-1}$ is samplable by a circuit of size $O(n\cdot 2^k)$ (and we can pick $2^k$ elements on which the first bit of the extractor is constant).  For affine sources, it can be shown by analyzing the $k$-th Gowers norm of the set of inputs on which the first bit of the extractor is constant (as pointed out to us by Ben Green).}

\item This is a rare case where a random function is {\em not} an optimal extractor.  For example, the parity function extracts one completely unbiased bit from any bit-fixing source with $k=1$ random bits, but we show that a random function will fail to extract from some such source with high probability.
\end{itemize}

Our first results concern explicit constructions of extractors for OBFS with $k$ sublogarithmic in $n$.

\begin{itemize}
\item We simplify and improve an explicit construction of extractors for OBFSs with small $k$ by Kamp and Zuckerman~\citep{KampZuckerman}.  In particular, the error parameter of our construction can be exponentially small in $k$, whereas the Kamp--Zuckerman construction achieves error that is polynomially small in $k$.  Our extractor (like that of \citep{KampZuckerman}) extracts only $\Theta(\log k)$ almost-uniform bits, in contrast to extractors for superlogarithmic $k$, which can extract nearly $k$ bits.

\item We prove that, when $k$ is sublogarithmic, the $\Theta(\log k)$ output length of our extractor is optimal for extractors for OBFSs computable by space-bounded streaming algorithms with a certain ``forgetlessness'' property. The class of streaming algorithms we analyze includes our construction as well as many natural random-walk based constructions. This is our main result.
\end{itemize}

Next, we investigate properties of random functions as extractors for OBFS's and find that $k \approx \log{n}$ appears to be a critical point for extractors for OBFSs in this setting as well.  Specifically, we show that:
\begin{itemize}
\item A random function is an extractor for OBFSs (with high probability) {\em if and only if} $k$ is at least roughly $\log n$.

\item In contrast, for the more relaxed concept of exposure-resilient functions, random functions suffice even for sublogarithmic $k$.  For static ERFs, $k$ can be as small as a constant, and for adaptive ERFs, $k$ can be as small as $\log\log n$.
\end{itemize}
All of the results concerning random functions yield resilient/exposure-resilient functions that output nearly $k$ almost-uniform bits.

\section{Preliminaries}

Throughout, we will use the convention that a lowercase number (e.g. $n$) implicitly defines a corresponding capital number ($N$) as its exponentiation with base $2$ (i.e. $N=2^n$).

\begin{definition}[Statistical Distance]
Let $\pr{X}$ and $\pr{Y}$ be two random variables taking values in a set $S$.  The \em statistical distance \em $\Delta(\pr{X},\pr{Y})$ between $\pr{X}$ and $\pr{Y}$ is
$$\Delta \left( \pr{X},\pr{Y} \right) = \max_{T \subset S}\left| \Pr \left[ \pr{X} \in T \right] - \Pr \left[ \pr{Y} \in T \right]\right| = \frac{1}{2}\sum_{w \in S} \left| \Pr \left[ \pr{X} = w \right] - \Pr \left[ \pr{Y} = w \right] \right|$$
\end{definition}
We will write $\pr{X} \approx_\ep \pr{Y}$ to mean $\Delta(\pr{X},\pr{Y}) \leq \ep$, and we will use $\Un{n}$ to denote the uniform distribution on $\B^n$. When $\Un{n}$ appears twice in the same set of parentheses, it will denote the same random variable. For example, a string chosen from the distribution $\dpr{\Un{n}}{\Un{n}}$ will always be of the form $w \circ w$ for some $w \in \B^n$. Note that $\dpr{\Un{n}}{\Un{m}}$ still equals $\Un{n+m}$.

\begin{definition}[Oblivious Symbol-Fixing Source]
An $(n,k,d)$ \em oblivious symbol-fixing source (OSFS) \em $\pr{X}$ is a source consisting of $n$ symbols, each drawn from $[d]$, of which all but $k$ are fixed and the rest are chosen independently and uniformly at random.
\end{definition}

\begin{definition}[Oblivious Bit-Fixing Source]
An $(n,k)$ \em oblivious bit-fixing source (OBFS) \em is an $(n,k,2)$ oblivious symbol-fixing source.
\end{definition}

We will use $\subsets{n}{\ell}$ to denote the set $\{L \subset [n] \colon |L| = \ell\}$ and, given some $L \in \subsets{n}{\ell}$ and a string $a \in \B^{\ell}$, we will write $L^{a,n}$ to denote the oblivious bit-fixing source that has the bits with positions in $L$ fixed to the string $a$.

\begin{definition}[Deterministic Randomness Extractor] \label{def:extractor}
Let $\mathcal{C}$ be a class of sources on $\B^n$.  A \em deterministic $\ep$-extractor \em for $\mathcal{C}$ is a function $E \colon \B^n \rightarrow \B^m$ such that for every $\pr{X} \in \mathcal{C}$ we have $E(\pr{X}) \approx_\ep \Un{m}$.
\end{definition}

Here we will focus mainly on deterministic randomness extractors for oblivious bit-fixing sources, also known as \em resilient functions \em (RFs).
\begin{definition}[Resilient Function]
\label{def:RF}
A \em $(k,\ep)$-RF \em is a function $f \colon \B^n \rightarrow \B^m$ that is a deterministic $\ep$-extractor for $(n,k)$ oblivious bit-fixing sources.
\end{definition}
We can also characterize extractors for OBFSs by their ability to fool a distinguisher: consider a computationally unbounded adversary $A$ that can set some of $f$'s input bits in advance but must allow the rest to be chosen uniformly at random. Then $f$ satisfies Definition~\ref{def:RF} if and only if $A$ is unable to distinguish between $f$'s output and the uniform distribution regardless of how $A$ changes $f$'s input.

When viewed through this lens, the notion of deterministic extraction from OBFSs has a natural relaxation obtained by restricting $A$ to only \em read \em (rather than modify) a portion of $f$'s input bits. Functions that are able to fool adversaries of this type are called \em exposure-resilient functions \em (ERFs). We define below the two simplest variants of exposure-resilient functions, which correspond to whether $A$ reads the bits of $f$'s input all at once or one at a time.

\begin{definition}[Static Exposure-Resilient Function]
A \em static $(k,\ep)$-ERF \em is a function $f \colon \B^n \rightarrow \B^m$ with the property that for every $L \in \subsets{n}{n-k}$, $f$ satisfies $\dpr{\Un{n}|_L}{f(\Un{n})} \approx_\ep \dpr{\Un{n}|_L}{\Un{m}}$.
\end{definition}

This definition can be restated in terms of average-case extraction using the following lemma, whose proof can be found in~\citep{Reshef}.
\begin{lemma}
\label{lem:ERFalternate}
A function $f \colon \B^n \rightarrow \B^m$ is a static $(k,\ep)$-ERF if and only if for every $L \in \subsets{n}{n-k}$, $f$ satisfies
$$\Exp_{a \leftarrow \Un{n-k}} {\left[ \Delta\left(f\left(L^{a,n}\right), \Un{m}\right) \right]} \leq \ep$$
\end{lemma}

Allowing the adversary to adaptively request bits of $f$'s input one at a time gives rise to the strictly stronger notion of an \em adaptive ERF\em:

\begin{definition}[Adaptive Exposure-Resilient Function]
An \em adaptive $(k, \ep)$-ERF \em is a function $f \colon \B^n \rightarrow \B^m$ with the property that for every algorithm $A \colon \B^n \rightarrow \B^*$ that can (adaptively) read at most $n-k$ bits of its input,\footnote{In other words, $A$ is a binary decision tree of depth $n-k-1$ with leaves labelled by its output strings and each internal node labelled by the position of the bit that $A$ requests at that juncture.} $f$ satisfies $\dpr{A(\Un{n})}{f(\Un{n})} \approx_\ep \dpr{A(\Un{n})}{\Un{m}}$.
\end{definition}

The following lemma will allow us to restrict our attention to algorithms $A$ that simply output the values of the bits that they request as they receive them (rather than outputting some function of those bits).

\begin{lemma}
\label{lem:AERFsimplification}
Let $A \colon \B^n \rightarrow \B^*$ be an adaptive adversary that reads at most $d$ bits of its input and let $A_r \colon \B^n \rightarrow \B^*$ be the algorithm that adaptively reads the same bits as $A$ and outputs them in the order that they were read.  For every function $f \colon \B^n \rightarrow \B^m$, the statistical distance between $\dpr{A(\Un{n})}{f(\Un{n})}$ and $\dpr{A(\Un{n})}{\Un{m}}$ is at most the distance between $\dpr{A_r(\Un{n})}{f(\Un{n})}$ and $\dpr{A_r(\Un{n})}{\Un{m}}$.
\end{lemma}
\begin{proof}
First, modify $A_r$ by padding its output with $0$'s so that its output length is always $d$.  Now define a second algorithm $A_p \colon \B^d \rightarrow \B^*$ as follows: on an input $x \in \B^d$, $A_p$ runs $A$, sequentially feeding it the bits of $x$ in response to $A$'s requests, and then outputs $A$'s output.  The fact that $A = A_p \circ A_r$ then implies the desired result.
\end{proof}

\section{A simplification and a lower bound}

In this section, we prove that when the entropy parameter $k$ is sublogarithmic in the input length $n$, an output length of $O(\log{k})$ is optimal for a natural class of space-bounded streaming algorithms, including algorithms that use the input bits to conduct a random walk on a graph. Before we state this lower bound, we give a simple improvement on the state of the art in explicit constructions of extractors for oblivious bit-fixing sources (i.e. resilient functions) for sublogarithmic entropy. Our lower bound then shows that the parameters achieved by this construction are optimal.

\subsection{The simplification}
We start with a simplification of a previous construction due to~\citep{KampZuckerman}.  The previous construction is based on very good extractors for oblivious \em symbol-fixing \em sources with $d \geq 3$ symbols obtained by using the symbols of the input string to take a random walk on an expander graph of degree $d$.  Since expander graphs do not exist with degree $d=2$, this approach could not be used for oblivious bit-fixing sources.  However, the construction of \citep{KampZuckerman} uses the fact that while a random walk on an expander is not an option, a random walk on a cycle still extracts some randomness even when the entropy $k$ of the input is very small.  Our construction is a slight modification of this random walk that simplifies the argument and improves the error parameter.

\begin{theorem} \label{thm:KZ}
 For every $n \in \N$, $k \in [n]$, $\ep > 0$, and $m = \frac{1}{2}( \log{k} - \log\log{(1/\ep)} )$, the function $f \colon \B^n \rightarrow \B^m$ defined by
 $$f(w) = \sum_{i=1}^n{w_i} \pmod{2^m}$$ is a $(k,\ep)$-RF. In particular, setting $\ep = 2^{-\sqrt{k}}$ gives output length $m = \frac{1}{4}\log{k}$.
\end{theorem}
\begin{proof}
 We can treat $f$ as computing the endpoint of a walk on $\Z{M}$ (where $M=2^m$) that starts at $0$ and either adds $1$ or $0$ to its state with every bit that it reads.  Since the endpoint of this walk does not depend on the order in which the input bits are processed, we may assume without loss of generality that all of the fixed bits in $f$'s input come at the beginning.  These bits only change the starting vertex of the random walk and do not affect the distance from uniform of the resulting distribution.  Therefore, to bound the distance from uniform of any distribution of the form $f(L^{*,n})$ we need only bound the mixing time of a walk on $\Z{M}$ consisting of $k$ random steps.  The following claim, whose proof we defer to the appendix, accomplishes this.

 \begin{claim}
 \label{claim:Fourier}
  Let $W_k$ be the distribution on the vertices of $\Z{M}$ (where $M=2^m$) obtained by beginning at $0$ and adding $1$ or $0$ with equal probability $k$ times.  The distance from uniform of $W_k$ is at most
  $$ \frac{e^{-k\pi^2/2M^2}}{2\left(1-e^{-3k\pi^2/2M^2}\right)} $$
 \end{claim}

Since $k \geq M^2$, the bottom of the fraction in Claim~\ref{claim:Fourier} is bounded from below by $2(1-e^{-3\pi^2/2}) > 1$ and so we have bounded the distance from uniform by $e^{-k\pi^2/2M^2}$.  With our setting of parameters this is at most $\ep^{\log{(e)}\pi^2/2} \leq \ep$, as desired.
\end{proof}

The difference between this construction and that of~\citep{KampZuckerman} is that each step of the random walk carried out by $f$ consists of adding either $1$ or $0$ rather than $1$ or $-1$ to the current state.  This has two advantages.  First, the random walk in the construction of~\citep{KampZuckerman} cannot be carried out on a graph of size $2^m$ since any even-sized cycle is bipartite and the walk traverses an edge at each step.  This necessitates an additional lemma about converting the output of the random walk to one that is almost uniformly distributed over $\B^m$, which incurs at error polynomially related to $k$.\footnote{This additional error was overlooked in~\citep{KampZuckerman}, and their Theorem $1.2$ erroneously claims an error exponentially small in $k$.}  By eliminating the need for this lemma, the construction of Theorem~\ref{thm:KZ} manages to achieve an exponentially small error parameter.  Second, setting $m=1$ in the construction of Theorem~\ref{thm:KZ} makes it clear that the idea underlying both it and the~\citep{KampZuckerman} construction is simply a generalization of bitwise addition modulo $2$---the parity function---which extracts $1$ uniformly random bit whenever $k \geq 1$.

As discussed previously, this construction achieves output length only logarithmic in $k$. This is considerably worse than the output length of $k - 2\log{(1/\ep)} - O(1)$ which we show to be possible both for extractors for OBFSs with $k > \log{n}$ (Section~\ref{sec:ProbabilisticMethodRFs}) and for ERFs (Section~\ref{sec:ProbabilisticMethodERFs}). The lower bound we prove in the following section shows why this is the case.

\subsection{The lower bound}
The extractor of Theorem~\ref{thm:KZ} is a symmetric function; that is, its output is not sensitive to the order in which the input bits are arranged. We begin building our more general negative result by first showing that extractors for OBFSs with this property cannot have superlogarithmic output length.

\begin{lemma} \label{lem:symmetric}
 Suppose that $X = L^{a,n}$ is an $(n,k)$-OBFS and that $f \colon \B^n \rightarrow \B^m$ is a symmetric function of the input bits in $[n] - L$.  (That is, for every permutation $\pi \colon [n] \rightarrow [n]$ that fixes $L$, $f(x_{\pi(1)}, \ldots, x_{\pi(n)}) = f(x_1,\ldots,x_n)$.)  Then $f(X) \approx_\ep \Un{m}$ implies that $m \leq \log{(k/(1-\ep))}$.
\end{lemma}
\begin{proof}
 By the symmetry of $f$ on the bits in $[n] - L$, the size of the support of $f(\pr{X})$ is at most $k$.  (The output depends only on the number of input bits in $[n] - L$ that equal $1$.)  Thus, the distance between $f(\pr{X})$ and $\Un{m}$ is at least $(M-k)/M$.  Together with $f(X) \approx_\ep \Un{m}$, this implies that $\ep \geq (M-k)/M$, which is equivalent to $m \leq \log{(k/(1-\ep))}$.
\end{proof}

We can use Lemma~\ref{lem:symmetric} to show that no symmetric function with large output length can be even a static ERF.

\begin{proposition} \label{prop:symmetric}
 If a symmetric function $f \colon \B^n \rightarrow \B^m$ is a static $(k,\ep)$-ERF then $m \leq \log{(k/(1-\ep))}$.
\end{proposition}
\begin{proof}
 From Lemma~\ref{lem:ERFalternate}, we have that for $f$ to be a static ERF, it must satisfy, for all sets $L \in \subsets{n}{n-k}$,
$$\Exp_{a \leftarrow \Un{n-k}} {\left[ \Delta\left(f\left(L^{a,n}\right), \Un{m}\right) \right]} \leq \ep$$
It follows by averaging that there exists a set $L$ and a string $a$ such that $f(L^{a,n}) \approx_\ep \Un{m}$.  Application of Lemma~\ref{lem:symmetric} to the source $L^{a,n}$ then yields the result.
\end{proof}

Since every deterministic $\ep$-extractor for $(n,k)$-OBFSs is a static $(k,\ep)$-ERF and every adaptive $(k,\ep)$-ERF is also a static $(k,\ep)$-ERF, Proposition~\ref{prop:symmetric} applies to extractors for OBFSs and adaptive ERFs as well.  Thus, 
Proposition~\ref{prop:symmetric} shows that constructions like that of Theorem~\ref{thm:KZ} and that of~\citep{KampZuckerman} are optimal.

However, there are many natural candidates for extraction from OBFSs that are similar to that of Theorem~\ref{thm:KZ} but are not symmetric, such as the analogous random walk on a directed version of a $3$-regular or $4$-regular expander graph. For instance, we could try the graph with vertex set $\mathbb{F}_p$ where the edge labelled $0$ from vertex $x$ goes to $x+1$ and the edge labelled $1$ goes to $x^{-1}$ (or $0$ in case $x = 0$). The undirected version of this graph is known to be an expander~\citep{Lubotzky}, so we might hope that with $k$ random steps we can reach an almost uniform vertex even for $p = 2^{\Omega(k)}$ and thus output $\Omega(k)$ almost-uniform bits.

$\mathbb{F}_p$ with inverse cords rather than an undirected cycle. It turns out that such constructions do no better, as we now show by extending the above lower bound for extractors for OBFSs to a large class of small-source streaming algorithms.  We start by defining the model of computation that we will assume.

\begin{definition}[Streaming Algorithm]
A \em streaming algorithm \em $A \colon \B^n \rightarrow \B^m$ is given by a $5$-tuple $(V,v_0,\Sigma^0,\Sigma^1,\varphi)$, where $V$ is the state space, $v_0 \in V$ is the initial state, $\Sigma^0 = (\sigma_1^0, \ldots, \sigma_n^0)$ and $\Sigma^1 = (\sigma_1^1, \ldots, \sigma_n^1)$ are two sequences of functions from $V$ to itself, and $\varphi$ is a function from $V$ to $\B^m$.  On an input sequence $(b_1,\ldots,b_n) \in \B^n$, $A$ computes by updating its state using the rule $v_{i+1} = \sigma_i^{b_i}(v_i)$.  $A$'s output is $A(b_1,\ldots,b_n) = \varphi(v_n)$.  The function $\varphi$ is called the \em output function \em of $A$, and the \em space \em of $A$ is $\log{|V|}$.

We say that $A$ is \em forgetless \em if and only if for every $i$ at least one of either ${\sigma_i}^0$ or ${\sigma_i}^1$ is a permutation.  (Thus, if the $i$-th bit is fixed to a certain value, $A$ does not ``forget'' anything about its state when reading that bit.)
\end{definition}

Forgetless streaming algorithms include random walks on $2$-regular digraphs that are consistently labelled (meaning that the edges labelled $b$ form a permutation, for each $b \in \B$), like the graph on $\mathbb{F}_p$ mentioned above. However, forgetless streaming algorithms are more general in the sense that they can compute random walks in which each step of the walk is conducted on a different graph.

We now show that forgetless streaming algorithms with small space cannot compute extractors for OBFSs with large output length (for small $k$). This is our main result.

\begin{theorem}
\label{thm:result}
 Suppose that $f \colon \B^n \rightarrow \B^m$ is a deterministic $\ep$-extractor for $(n,k)$-OBFSs that can be computed by a forgetless streaming algorithm with space $s \leq \log{(n/k)}/k$.  Then $m \leq \log{(k/(1-\ep))}$.
\end{theorem}
\begin{proof}[Proof]
 Fix an $\ep$-extractor for $(n,k)$-OBFSs $f \colon \B^n \rightarrow \B^m$ and let $A$ be a forgetless streaming algorithm with space $s \leq \log{(n/k)}/k$ that computes $f$.  To show that $m \leq \log{(k/(1-\ep))}$, we will first reduce to a special case in which we can make some simplifying assumptions about $A$.  We will then construct an oblivious bit-fixing source $\pr{X}$ such that $f$ is symmetric on the set of bit positions not fixed by $\pr{X}$.  This will allow us to apply Lemma~\ref{lem:symmetric} to obtain our result since $f$ must map $\pr{X}$ close to uniform.
 \\
 \textit{Reduction to the special case:} Let $\Sigma^0$ and $\Sigma^1$ be the sequences of functions used by $A$, and let $\varphi$ be its output function.  We reduce to the special case that every element of $\Sigma^0$ is the identity.

 Since $A$ is forgetless, we can switch some of the functions $\sigma_i^0$ and $\sigma_i^1$ to make every function in $\Sigma^0$ a permutation while preserving the fact that $A$ computes a $(k,\ep)$-RF.  (This corresponds to just negating some input bits.)  This allows us to define a new sequence of functions $F = \{f_1,\ldots, f_n\}$ and a new output function $\psi$ by the following relations.
 \begin{eqnarray*}
 \sigma_i^0 \circ \cdots \circ \sigma_1^0 \circ f_i &=& \sigma_i^1 \circ \sigma_{i-1}^0 \circ \cdots \circ \sigma_1^0\\
 \psi &=& \varphi \circ \sigma_n^0 \circ \cdots \circ \sigma_1^0
 \end{eqnarray*}
 Then $(V, v_0, (\mbox{id},\mbox{id},\ldots, \mbox{id}), (f_1, \ldots, f_n), \psi)$ can be verified to be a streaming algorithm that computes the same function as $(V, v_0, \Sigma^0, \Sigma^1, \varphi)$.\\
 \\
 \textit{Constructing the source $X$:} Letting $S = 2^s$, we can choose a set $F_1 \subset F$ of size at least $n/S$ such that all the functions in $F_1$ map the initial state $v_0$ to some common state (call it $v_1$).  We can then choose a set $F_2 \subset F_1$ of size at least $n/S^2$ such that all functions in $F_2$ map $v_1$ to some common state, which we call $v_2$.  Continuing in this way, we obtain a set $F_k \subset F$ of size at least $n/S^k$ and a sequence $(v_0,\ldots,v_k)$ with the property that every $f \in F_k$ satisfies $f(v_i) = v_{i+1}$ for $0 \leq i < k$.  We now define $\pr{X}$ to be the oblivious bit-fixing source that has the bits at positions that correspond to functions in $F_k$ un-fixed and the rest of the bits fixed to $0$. By our assumption that $s \leq \log{(n/k)}/k$, we have $|F_k| \geq n/S^k \geq k$, meaning that $\pr{X}$ has at least $k$ unfixed bits.  \\
 \\
 \textit{Obtaining the desired bound:} For any string $w$ in the support of $\pr{X}$, $f$'s output will be $\psi(v_{H(w)})$ where $H(w)$ is the Hamming weight of $w$.  Therefore $f$ is a symmetric function of the bits in positions not fixed by $\pr{X}$.  Since $\pr{X}$ contains at least $k$ independent, uniformly random bits and $f$ is a $(k,\ep)$-resilient function, Lemma~\ref{lem:symmetric} yields $m \leq \log{(k/(1-\ep))}$ as desired.
\end{proof}

What does this theorem tell us about extraction in low-entropy settings? If we set $s = m \leq k$ (as in the walk on the cycle of Theorem~\ref{thm:KZ}) then Theorem~\ref{thm:result} implies that when $k < \sqrt{\log{n} - \log\log{n}}$ we are confined to output length $m \leq \log{( k/( 1-\ep ) )}$.  In other words, the output length of $\Omega(\log{k})$ offered by Theorem~\ref{thm:KZ} is close to optimal for extractors in this model when $k < \sqrt{\log{n}}$.

We note here a separate, trivial space lower bound that applies even to the forgetful case: since streaming algorithms under our model cannot produce any output bits until they have read all the input bits, we have $s > m - 1$ when $\ep < 1/2$. This bound can in fact be generalized to streaming algorithms that are allowed to output bits at any point in their computation by a simple adaptation of a space lower bound for strong extractors proven in~\citep{BarYossefTrevisanReingoldShaltiel}. The resulting lower bound says that $s \geq m - 4$ when $\ep \leq 1/8$ and $k \leq n/2$ for extractors for OBFSs computable by any streaming algorithm.

\section{Non-constructive results}

We now turn to determining for what values of the entropy parameter $k$ it is possible to achieve output length $m = \Omega(k)$ using the probabilistic method. Here we find that the results are roughly in agreement with our explicit lower bounds from the previous section. That is, a randomly chosen function $f \colon \B^n \rightarrow \B^m$ will almost always be an extractor for OBFSs with output length $m = \Omega(k)$ when $k$ is larger than $\log{n}$, and this output length cannot be achieved using the probabilistic method when $k < \log{n}$.

We then show that random functions can do better in the more relaxed realm of exposure-resilient functions: a randomly chosen function is almost always a static ERF with optimal output length for any $k$, and an \em adaptive \em ERF with optimal output length when $k$ is larger than $\log\log{n}$.

Before we proceed, we state a Chernoff bound and a partial converse to it that we will use in proving these results. A sketch of the proof of Lemma~\ref{lem:ChernoffConverse} is given in the appendix.

\begin{lemma}[A Chernoff bound]
\label{lem:Chernoff}
Let $\pr{X_1}, \ldots, \pr{X_t}$ be independent random variables taking values in $[0,1]$, and let $\pr{X} = ( \sum_i{\pr{X_i}} )/t$ and $\mu = \Exp[\pr{X}]$.  Then for every $0 < \ep < 1$, we have
$$\Pr \left[ \left| \pr{X} - \mu \right| > \ep \right] < 2e^{-t\ep^2/2} \leq 2^{-\lfloor \Omega\left(t\ep^2\right) \rfloor}$$
\end{lemma}

\begin{lemma}[Partial converse of Chernoff bound]
\label{lem:ChernoffConverse}
Let $\pr{X_1}, \ldots, \pr{X_t}$ represent the results of independent, unbiased coin flips, and let $\pr{X} = ( \sum_i{\pr{X_i}} )/t$.  Then for every $0 \leq \ep \leq 1/2$, we have
$$\Pr \left[ \left| \pr{X} - \frac{1}{2} \right| \geq \ep \right] \geq 2^{-\lceil O\left(t\ep^2\right) \rceil}$$
\end{lemma}

\subsection{Deterministic extractors for OBFSs}
\label{sec:ProbabilisticMethodRFs}

Theorem~\ref{thm:RFexistence} below, which follows from a straightforward application of the Chernoff bound stated in Lemma~\ref{lem:Chernoff}, shows that the probabilistic methods gives extractors for OBFSs with $k > \log{n}$. Theorem~\ref{thm:RFnonexistence} then shows that $k > \log{n}$ is the best we can do using the probabilistic method.

\begin{theorem}
\label{thm:RFexistence}
 For every $n \in \N$, $k \in [n]$, and $\ep > 0$, a randomly chosen function $f \colon \B^n \rightarrow \B^m$ with $m \leq k - 2\log{(1 / \ep)} - O(1)$ and $k \geq \max\lbrace\log{(n-k)},\log\log{\binom{n}{k}}\rbrace + 2\log{(1/\ep)} + O(1)$ is a deterministic $\ep$-extractor for $(n,k)$-OBFSs with probability at least $1 - 2^{-\Omega(K\ep^2)}$, where $K = 2^k$.
\end{theorem}
\begin{proof}
Fix an $(n,k)$-OBFS $\pr{X}$. Choosing the function $f$ consists of independently assigning a string in $\B^m$ to each string in the support of $\pr{X}$. In order for $f$ to map $\pr{X}$ close to uniform, we need to have chosen it such that, for every fixed statistical test $T \subset \B^m$, the fraction of strings in $\pr{X}$ mapped by $f$ into $T$ is very close to the density of $T$ in $\B^m$. This is expressed formally by the condition below.
$$\left| \frac{|f^{-1}(T)|}{2^k} - \frac{|T|}{2^m} \right| \leq \ep$$
Now fix one specific test $T \subset \B^m$. For each string $w$ in the support of $\pr{X}$, define the indicator variable $I_w$ to be $1$ if $f(w) \in T$ and $0$ otherwise. Then Lemma~\ref{lem:Chernoff} (our Chernoff bound) applied to $\left( \sum_w{I_w} \right)/2^k = |f^{-1}(T)|/2^k$ shows that $f$ fails the condition above with probability at most $2^{-\lfloor \Omega(K\ep^2) \rfloor}$.

There are $2^M$ possible tests $T \subset \B^m$ (where $M = 2^m$). A union bound over all these tests therefore gives that the probability that $f$ fails to map $\pr{X}$ to within $\ep$ of uniform is at most $2^{M-\lfloor \Omega(K\ep^2) \rfloor}$. We can perform a similar union bound over the possible choices of the source $\pr{X}$: there are $\binom{n}{k}N/K$ such sources, yielding that the probability that $f$ is not a $(k, \ep)$-RF is at most
$$\binom{n}{k}\frac{N}{K} \hspace{.05in} 2^{M-\lfloor \Omega\left(K\ep^2\right) \rfloor} = 2^{-\Omega\left(K\ep^2\right)}$$
provided $K \geq \max\{\log{(\frac{N}{K})}, \log{\binom{n}{k}}\} c/ \ep^2$ for a sufficiently large constant $c$ and $M \leq c\ii K\ep^2$ for a sufficiently small constant $c\ii$.  Taking logarithms gives the result.
\end{proof}

The $\max\{\log{(n-k)},\log\log{\binom{n}{k}}\}$ term in the statement of Theorem~\ref{thm:RFexistence} is always at most $\log{n}$, so the theorem always holds when $k \geq \log{n} + 2\log{(1/\ep)} + O(1)$, as discussed earlier. In the following theorem, we prove a limitation on the extraction properties of random functions which shows that this bound on $k$ is in fact nearly tight.

\begin{theorem}
\label{thm:RFnonexistence}
 There is a constant $c$ such that for every $n \in \N$, $k \in [n]$, and $\ep \in [0, 1/2]$ satisfying $k \leq \log{(n-k)} + 2\log{(1/\ep)} - c$, a random function $f \colon \B^n \rightarrow \B$ will fail to be a deterministic $\ep$-extractor for $(n,k)$-OBFSs with probability at least $1 - 2^{-\sqrt{N/K}}$, where $N = 2^n$ and $K = 2^k$.
\end{theorem}
\begin{proof}
 Fix an input size $n$ and a set $L$ of $n-k$ fixed bits (say, $L = [n-k]$).  To say that $f$ an $\ep$-extractor for $(n,k)$-OBFSs is to say that all $2^{n-k}$ sets $S$ of the form $L^{*,n}$ satisfy the following condition.
$$\left|\Pr_{w \leftarrow S} \left[ f(w) = 1 \right] - \frac{1}{2} \right| \leq \ep$$
Since $f(w)$ is chosen independently for each string $w \in S$, we can use the converse of our Chernoff bound (Lemma~\ref{lem:ChernoffConverse}) to say that the probability that $f$ satisfies this condition for a fixed set $S$ is at most $1 - 2^{-\lceil O(K\ep^2) \rceil}$, where $K = 2^k = |S|$.

Since there are $N/K$ subsets of the form $L^{*,n}$ and they are disjoint, the probability that $f$ will fail the above condition on none of them (i.e. the probability that $f$ is a resilient function) is at most
$$\left( 1-2^{- \lceil O \left( K\ep^2 \right) \rceil} \right)^{N/K}$$
If the $O \left( K\ep^2 \right)$ term is less than or equal to $1$, this probability is at most $2^{-N/K}$. Otherwise, it is at most  $2^{-\sqrt{N/K}}$ provided that $N/K \geq 2^{CK\ep^2}$ for a sufficiently large constant $C = 2^c$.  Taking logarithms twice completes the proof.
\end{proof}

Theorem~\ref{thm:RFnonexistence} does not establish that extractors for OBFSs with the stated parameters do not exist; indeed, as mentioned earlier, the parity function (i.e. $f(x_1,\ldots,x_n) = \oplus{x_i}$) is a perfect resilient function for even $k=1$. What the theorem does show, however, is that $k \approx \log{n}$ represents a critical point below which these extractors become very rare. This seems consistent with the lower bound on $k$ proven in Theorem~\ref{thm:result}.

\subsection{Exposure-resilient functions}
\label{sec:ProbabilisticMethodERFs}

We now show that probabilistically constructing exposure-resilient functions is easier than constructing extractors for OBFSs. This is because, while the adversary can choose input sources in the extractor setting, here it can only expose them. The probabilistic constructions of static and adaptive ERFs both proceed by counting the number of adversaries that must be fooled and then applying Lemma~\ref{lem:foolOneAdversary} (below), which is an upper bound on the probability that a randomly chosen function will fail to fool a fixed adversary.  This lemma applies equally both to static and adaptive adversaries; the difference in achievable parameters between static and adaptive ERFs therefore stems solely from the fact that there are many more adversaries in the adaptive setting.

\begin{lemma}
\label{lem:foolOneAdversary}
 Let $A \colon \B^n \rightarrow \B^*$ be an algorithm that reads at most $d$ bits of its input, let $\ep > 0$, and choose a function $f \colon \B^n \rightarrow \B^m$ uniformly at random with $m = n-d - 2\log{(1/\ep)} - O(1)$.  Then $f$ will fail to satisfy
 $$\dpr{A\left(\Un{n}\right)}{f\left(\Un{n}\right)} \approx_\ep \dpr{A\left(\Un{n}\right)}{\Un{m}}$$
 with probability at most $2^{-\Omega(N\ep^2)}$, where $N = 2^n$.
\end{lemma}

\begin{proof}
Lemma~\ref{lem:AERFsimplification} allows us to assume without loss of generality that $A$ adaptively reads $d$ bits and outputs them in the order that they were read.  Under this assumption, we have $\dpr{A(\Un{n})}{\Un{m}} = \Un{d+m}$.  We therefore need only to bound the probability that $\dpr{A(\Un{n})}{f(\Un{n})}$ is far from $\Un{d+m}$.

Fix a statistical test $T \subset \B^d \times \B^m$.  In order for $\dpr{A(\Un{n})}{f(\Un{n})}$ to pass this specific test of uniformity, we need $f$ to satisfy
\begin{equation} \label{passOneTest}
\left| \Pr \left[ \dpr{A\left(\Un{n}\right)}{f\left(\Un{n}\right)} \in T \right] - \frac{|T|}{2^{d+m}} \right| \leq \ep
\end{equation}
For every $w \in \B^n$, define $I_w$ to be $1$ if $(A(w),f(w)) \in T$ and $0$ otherwise, and notice that $\Pr [ \dpr{A(\Un{n})}{f(\Un{n})} \in T ] = \frac{1}{2^n}\sum_w{I_w}$.  For $x \in \B^d$, let $T_x$ denote $T \cap ( \{x\} \times \B^m )$.  Then, for a fixed $w$, the expectation of $I_w$ over the choice of $f$ is exactly $|T_{A(w)}| / 2^m$, and so by the regularity of $A$ the expectation of $\frac{1}{2^n}\sum_w{I_w}$ over the choice of $f$ is $|T|/2^{d+m}$.  A Chernoff bound (Lemma~\ref{lem:Chernoff}) then gives that the probability over the choice of $f$ that Equation~(\ref{passOneTest}) is not satisfied is at most $2^{-\lfloor \Omega(N\ep^2) \rfloor}$.

Since there are $2^{DM}$ possible choices of $T$ in the above analysis (where $D = 2^d$, $M=2^m$), a union bound shows that the probability that $\dpr{A(\Un{n})}{f(\Un{n})}$ will fail one or more of them is at most $2^{DM}2^{-\lfloor \Omega(N\ep^2) \rfloor} = 2^{-\Omega(N\ep^2)}$ if $m = n - d - 2\log{(1/\ep)} - c$ for a sufficiently large constant $c$.
\end{proof}

Having established that a random function will tend to fool a fixed adversary, we now establish the existence of static and adaptive exposure-resilient functions. In both cases, we do so by taking a union bound over all potential adversaries and applying Lemma~\ref{lem:foolOneAdversary}. Thus, the parameters achieved are those that bring the number of adversaries to below $2^{N\ep^2}$.

\begin{theorem}
\label{thm:staticERFexistence}
 For every $n \in \N$, $k \in [n]$, and $\ep \geq c\sqrt{n/2^n}$ where $c$ is a universal constant, a randomly chosen function $f \colon \B^n \rightarrow \B^m$ with $m \leq k - 2\log{(1/\ep)} - O(1)$ is a static $(k,\ep)$-ERF with probability at least $1 - 2^{-\Omega(N \ep^2)}$, where $N = 2^n$.
\end{theorem}
\begin{proof}
 Every static adversary that tries to distinguish the output of $f$ from uniform is an algorithm $A \colon \B^n \rightarrow \B^{n-k}$ that reads exactly $n-k$ bits of its input.  We can therefore apply Lemma~\ref{lem:foolOneAdversary} with $d = n-k$ to get that the probability that $f$ will fail to fool any one adversary is at most $2^{-\Omega(N\ep^2)}$.  Taking a union bound over the $\binom{n}{k}$ possible adversaries, we get that the probability that $f$ will not fool all adversaries is at most
$$ \binom{n}{k}2^{-\Omega\left(N\ep^2\right)} \leq N 2^{-\Omega\left(N\ep^2\right)} = 2^{-\Omega\left(N\ep^2\right)}$$
where the final equality is given by the constraint on $\ep$.
\end{proof}

Counting the number of adversaries in the adaptive setting is a bit more work, but Lemma~\ref{lem:AERFsimplification} from our preliminaries simplifies this task.
\begin{theorem}
\label{thm:adaptiveERFexistence}
 For every $n \in \N$, $k \in [n]$, and $\ep > 0$, a randomly chosen function $f \colon \B^n \rightarrow \B^m$ with $m \leq k - 2\log{(1 / \ep)} - O(1)$ and $k \geq \log{\log{n}} + 2\log{(1 / \ep)} + O(1)$ is an adaptive $(k,\ep)$-ERF with probability at least $1 - 2^{-\Omega(N \ep^2)}$, where $N = 2^n$.
\end{theorem}
\begin{proof}
 The proof is identical to that of Theorem~\ref{thm:staticERFexistence} except that we have to count the number of adaptive adversaries.  We do so below.

 First we note that Lemma~\ref{lem:AERFsimplification} implies that if $f$ fools all adaptive adversaries that output the bits they read as they read them, then $f$ fools all adaptive adversaries.  We therefore only need to count this smaller set of adversaries.  The process by which such an adversary chooses which bits to request can be modelled by a decision tree of depth $n-k-1$ whose internal nodes are labelled by elements of $[n]$. Since the number of nodes in such a tree is $2^{n-k-1}-1 < N/2K$, where $N=2^n$ and $K=2^k$, we can bound the total number of trees---and therefore adversaries---by $n^{N/2K}$.

Proceeding with the same kind of union bound as in the proof of Theorem~\ref{thm:staticERFexistence}, we see that the probability that $f$ will not fool all adaptive adversaries is at most $n^{N/2K}2^{-\Omega(N\ep^2)} = 2^{-\Omega(N\ep^2)}$, provided that $K \geq (c\log{n})/\ep^2$ for a sufficiently large constant $c$.  Taking logarithms yields the theorem.
\end{proof}

\section{Future work}
 The general question of whether there exist resilient functions with large output length in the low-entropy range studied here is still unresolved.
 \begin{question}
 Does there exist, for all $n \in \N$ and some growing function $0 < k(n) < \log{n}$, a deterministic $\ep$-extractor for $(n, k(n)$-OBFSs with output length $m = \Omega(k(n))$ and $\ep$ constant?
 \end{question}
 Theorem~\ref{thm:result} shows that to resolve this question in the positive direction requires a function that is either not computable by a forgetless streaming algorithm or uses a considerable amount of space. In the other direction, an interesting step towards a negative result would be to at least remove the forgetlessness condition from the space lower bound proven in that theorem.

 We can ask an analogous question for the case of adaptive ERFs with $k < \log{\log{n}}$.
 \begin{question}
 Does there exist, for all $n \in \N$ and some growing function $0 < k(n) < \log{\log{n}}$, an adaptive $(k(n),\ep)$-ERF with output length $m = \Omega(k(n))$ and $\ep$ constant?
 \end{question}
 In this case, we cannot even rule out the possibility that a more clever use of the probabilistic method will resolve this question positively. Thus, a first step toward a negative result might be to prove an analogue to Theorem~\ref{thm:RFnonexistence} that shows that adaptive ERFs with near-optimal output length become very rare when $k < \log{\log{n}}$.

 A third open problem arising from this work is that of finding an explicit construction of a static ERF with the parameters achieved using the probabilistic method in Theorem~\ref{thm:staticERFexistence}. Currently, an output length of $\Omega(k)$ is achieved in~\citep{DodisSahaiSmith} using strong extractors, but the construction works only when $k > \log{n}$. For $k$ smaller than $\log{n}$, there is no known construction of a static ERF that is not also an RF, making the construction of Theorem~\ref{thm:result} the current state of the art. This leaves us with the following open question:
 \begin{question}
 Does there exist, for all $n \in \N$ and some growing function $0 < k(n) < \log{n}$, an explicit static $(k(n), \ep)$-ERF with output length $m = \Omega(k(n))$ and $\ep$ constant?
 \end{question}

\bibliographystyle{alphanum}
\bibliography{Bibliography}

\appendix
\section[Proof sketch of Lemma 3.2]{Proof sketch of Lemma~\ref{lem:ChernoffConverse}}
\begin{lemma*}
Let $\pr{X_1}, \ldots, \pr{X_t}$ represent the results of independent, unbiased coin flips, and let $\pr{X} = ( \sum_i{\pr{X_i}} )/t$.  Then for every $0 \leq \ep \leq 1/2$, we have
$$\Pr \left[ \left| \pr{X} - \frac{1}{2} \right| \geq \ep \right] \geq 2^{-\lceil O\left(t\ep^2\right) \rceil}$$
\end{lemma*}
\begin{proof}[\textit{Proof Sketch}]
We address three separate cases: $0 \leq \ep < \frac{1}{4\sqrt{t}}$, $\frac{1}{4\sqrt{t}} \leq \ep < \frac{1}{5}$, and $\frac{1}{5} \leq \ep \leq \frac{1}{2}$. In the first case, we upper-bound the probability that $\left| \pr{X} - \frac{1}{2} \right| < \ep$ using the fact that no term of the binomial distribution exceeds $\sqrt{2 / \pi t}$ in probability mass. In the second case, we set $\beta = \frac{1}{2} + 2\ep$ and use Stirling's approximation to lower-bound the probability by
\begin{eqnarray*}
\lfloor \ep t \rfloor \cdot \binom{t}{\lfloor \beta t \rfloor} / 2^t &\geq& \lfloor \ep t \rfloor \cdot
        \frac{1}{\sqrt{t}}
        \frac{1}{2^t \beta^{\beta t} (1-\beta)^{(1-\beta)t}} \\
        &\geq& 2^{-O(t\ep^2)}
\end{eqnarray*}
where the first inequality is from Stirling's approximation. In the third case, we just lower-bound the probability by $2^{-t}$.
\end{proof}

\section[Proof of Claim 4.2]{Proof of Claim~\ref{claim:Fourier}}
\begin{claim*}
  Let $W_k$ be the distribution on the vertices of $\Z{M}$ (where $M=2^m$) obtained by beginning at $0$ and adding $1$ or $0$ with equal probability $k$ times.  The distance from uniform of $W_k$ is at most
  $$ \frac{e^{-k\pi^2/2M^2}}{2\left(1-e^{-3k\pi^2/2M^2}\right)} $$
\end{claim*}
\begin{proof}
\Ynote{NEW MATERIAL}
Consider $\Z{M}$ as an additive group, and let $P$ be the probability distribution on $\Z{M}$ that equals $0$ with probability $1/2$ and $1$ otherwise. Then the distribution on $\Z{M}$ after $k$ steps of our random walk is $P^{*n}$, the $n$-th convolution of $P$ with itself.

Lemma 1 in Chapter 3 of~\citep{Diaconis} bounds the distance between $P^{*n}$ and the uniform distribution in terms of the traces of the Fourier transforms by $P^{*n}$ of the non-trivial irreducible representations of $\Z{M}$. This simplifies nicely since the Fourier transform $\widehat{P^{*n}}(\rho)$ of a representation $\rho$ by $P^{*n}$ equals $(\hat{P}(\rho))^n$, the $n$-th power of the Fourier transform of $\rho$ by $P$. Since there is one non-trivial irreducible representation for each $j \in [M-1]$, we therefore arrive at the following upper bound for the distance from uniform after $k$ random steps.
$$\frac{1}{4}\sum_{j=1}^{M-1}{\left( \frac{1}{2}+\frac{1}{2}\cos\left(\frac{2\pi j}{M}\right) \right)^k}$$

To bound this sum, we first note that $\frac{1}{2} + \frac{1}{2}\cos(x) \leq e^{-x^2/8}$ for $x \in [0, \pi]$. This, together with the fact that $M=2^m$ is even, allows us to write
\begin{eqnarray*}
 \frac{1}{4}\sum_{j=1}^{M-1}{\left( \frac{1}{2}+\frac{1}{2}\cos\left(\frac{2\pi j}{M}\right) \right)^k} &=& \frac{1}{2}\sum_{j=1}^{(M-2)/2}{\left( \frac{1}{2}+\frac{1}{2}\cos\left(\frac{2\pi j}{M}\right) \right)^k} \\
&\leq& \frac{1}{2}\sum_{j=1}^{(M-2)/2}{e^{-k\pi^2 j^2 /2M^2}} \\
&\leq& \frac{1}{2}e^{-k\pi^2/2M^2}\sum_{j=1}^{\infty}{e^{-k\pi^2 (j^2 - 1) /2M^2}} \\
&\leq& \frac{1}{2}e^{-k\pi^2/2M^2}\sum_{j=0}^{\infty}{e^{-3k\pi^2 j /2M^2}} \\
&=& \frac{e^{-k\pi^2/2M^2}}{2\left(1-e^{-3k\pi^2/2M^2}\right)}
\end{eqnarray*}
which is the desired result.
\end{proof}
\end{document}